\documentclass{article}
\usepackage{amsfonts}
\usepackage{enumerate}
\usepackage{url}
\usepackage{amsmath,amsthm,amssymb,amscd,epsfig}
\usepackage{graphicx, caption}
\usepackage{psfrag}
\usepackage{amsmath}
\usepackage{endnotes}
\usepackage{pdfpages}

\newcommand {\E}{\mathbb{E}}
\newcommand {\N}{\mathbb{N}}
\newcommand {\Z}{\mathbb{Z}}
\newcommand {\R}{\mathbb{R}}

\newcommand {\V}{\mathbb{V}}

\newcommand {\rd}{\mathrm{d}}

\newcommand {\T}{\mathbb{T}}

\newtheorem{theorem}{Theorem}[section]
\newtheorem{proposition}[theorem]{Proposition}

\newtheorem{lemma}[theorem]{Lemma}

\numberwithin{equation}{section}

\begin{document}
\title{A rigourous demonstration of the validity of \\Boltzmann's scenario for
the spatial homogenization of a freely expanding gas and the equilibration of the Kac ring}
\author{S. {\ De Bi\`{e}vre}\thanks{Stephan.De-Bievre@math.univ-lille1.fr} \\
Univ. Lille, CNRS, UMR 8524 - Laboratoire Paul Painlev\'e\\ F-59000 Lille, France\\
Equipe-Projet MEPHYSTO,
Centre de Recherche INRIA Futurs,\\
Parc Scientifique de la Haute Borne, 40, avenue Halley B.P. 70478,\\
F-59658 Villeneuve d'Ascq cedex, France\\
and\\
P.~E. {Parris}\thanks{
parris@mst.edu} \\
Department of Physics\\
Missouri University of Science \& Technology,\\
Rolla, MO 65409, USA\\
}
\date{January 14 2017}

\maketitle
\abstract{Boltzmann provided a scenario to explain why individual macroscopic systems composed of a large number $N$ of microscopic constituents are inevitably (i.e., with overwhelming probability) observed to approach a unique macroscopic state of thermodynamic equilibrium, and why after having done so, they are then observed to remain in that state, apparently forever. We provide here rigourous  new results that mathematically prove the basic features of Boltzmann's  scenario for two classical models: a simple boundary-free model for the spatial homogenization of a non-interacting gas of point particles, and the well-known Kac ring model. Our results, based on concentration inequalities that go back to Hoeffding, and which focus on the typical behavior of individual macroscopic systems, improve upon previous results by providing  estimates, exponential in $N$, of probabilities and time scales involved.}

\newpage
\section{Introduction}
It is a fundamental tenet of thermodynamics, well-grounded in common experience, that an individual
isolated macro\-scopic sys\-tem will, starting from an arbitrary initial state, evolve
irreversibly towards a unique stationary state of thermodynamic equilibrium, in
which it will then remain.
Boltzmann, assuming macroscopic matter to be composed of a large number $N$ of elementary
constituents (atoms and/or molecules), provided a compelling scenario to explain how this
irreversible behaviour arises from the manifestly reversible laws of mechanics. His
arguments, which have been detailed and re-explained by numerous authors~\cite{Eh59, Ka59, Fe67,
Le93, Br96, Pe90, Le08}, involve several crucial elements.

First, Boltzmann identified
states of thermodynamic equilibrium with those macrostates of a system
which correspond to the largest number of microstates that are actually
consistent with the physical constraints imposed upon it.

Next he argued that, because the microstates corresponding to
equilibrium occupy such an overwhelmingly large fraction of the available
phase space, such a macroscopic system would almost certainly evolve from any microstate that does
not correspond to equilibrium to one that does.

This scenario allowed Boltzmann to both explain the inevitably observed
approach to equilibrium,
and to refute objections to his ideas raised in the form of
paradoxes famously posed by Zermelo and by Loschmidt.

Crucial to Boltzmann's arguments is the fact
that the number of constituents making up macroscopic collections of matter is
overwhelmingly large.

This fact, in his view, strongly suggested that those initial microstates of an initially constrained  macroscopic system that do exhibit an approach to equilibrium when the constraints are removed (i.e., those that are, in fact, consistent with thermodynamics) are themselves overwhelmingly numerous, and therefore ``typical'' of the initial states in which such an initially constrained system is likely to be prepared in the absence of (perhaps extraordinary) measures taken to prevent it.

The enormously large number of constituents of macroscopic systems
should also, he suggested, allow for estimates to
be made of the enormously long interval of time over which a
typical individual system will appear to remain in the equilibrium
macrostate it eventually reaches.

Obviously, he reasoned, this equilibrium residence time must be
longer than the duration of humanly-possible observation times,
although necessarily shorter than the \emph{recurrence times},
also predicted by  the laws of mechanics, at which the system
must necessarily pass arbitrarily close to its initial state.

Thus, according to Boltzmann, it is the overwhelming fraction of typical initial
conditions and the enormity of typical equilibrium residence times
that allow the aforementioned paradoxes to be resolved.

Despite these compelling but generally unproven arguments, Boltzmann's
scenario continues to meet resistance and various claims can be
found in the literature contesting their validity.
We refer to~\cite{Br96, Le08} for a detailed analysis of the continuing
controversy. As an example, although it has been thoroughly refuted, it is still commonly argued that
ergodicity or mixing is both a necessary
and sufficient ingredient for a macroscopic dynamical system to demonstrate an approach to
equilibrium. This situation is perhaps a consequence of the fact that a full mathematical and even physical treatment of Boltzmann's ideas is not
yet available~\cite{Ka59, Ru91, Le08, Vi14}.

It is the view of the authors that a useful step in clearing up the controversy that
persists regarding the approach to equilibrium and the irreversibility of individual
macroscopic systems, is the identification of deterministic time-reversible
dynamical models for which Boltzmann's scenario can be
pushed through with complete mathematical rigour. Such models can help
to eliminate all doubts as to what, precisely, is being claimed and how, exactly,
the various longstanding objections to
Boltzmann's scenario come to be resolved.

To this end we
revisit here two classical models that allow such a
program to be fully and explicitly carried out in the context of the approach
to equilibrium of individual macroscopic systems.

We study first a very natural, simple, and therefore tractable model \cite{Fr58, Be10, Be12} for what
is arguably the most common
textbook example of approach to equilibrium: the free expansion
of a gas. The microscopic dynamics, assumed to be that of a gas of $N$ non-interacting point particles,
is Hamiltonian and time-reversible. To simplify the situation further, we envision the
dynamics of the gas particles as taking place on a $d$-dimensional torus, so that there are no
interactions, even, with the walls of any hypothetical container (See Fig. 1). As a result, the dynamics of
the $N$-particle system is recurrent and neither mixing nor ergodic on the energy surface. It is actually completely integrable.
Since there are no energy exchanges between the particles themselves,
and no interactions with any environment, the gas in our model can obviously not approach
(although it will maintain) a thermal distribution of particle velocities as it
expands. What we focus on here, therefore, for individual $N$-particle systems, is
the irreversible approach of an initially inhomogeneous coarse-grained gas density profile of independent particles
to a state of uniform spatial homogeneity (See Fig.~\ref{fig:expand}).

Our main result, stated as
 Theorem~\ref{thm:gasexpansion}, asserts that with overwhelmingly large probability,
to be made more precise in what follows,
the coarse-grained density profile of a typical individual $N$-particle system of this type will evolve towards a spatially homogeneous distribution and that, after having achieved such a spatially homogeneous macrostate, it will again be observed to be in
that state for a very long sequence of times, provided $N$ is sufficiently large.
We stress again the feature of our results that is crucial to fully implementing Boltzmann's scenario: they concern the evolution of typical individual $N$-particle systems, and not just the average behaviour associated with a statistical ensemble.  Our main result improves upon previous work by providing  estimates, exponential in $N$, of probabilities and time scales involved.
Moreover, as our proofs use little detailed information on the one-particle dynamics, our results generalize readily to other geometries, as explained in Section~\ref{s:examples}.

Another model that has been often used to illustrate the validity of Boltzmann's scenario is the
Kac ring model~\cite{Ka59, Th72, Br96,GoOl09, MaNeSh09}, in which $N$ black or white balls, in discrete time
steps,  move in the same direction around a ring, on which there are $m$ randomly-placed obstacles that
reverse the color of each ball that passes.
The microscopic dynamics of the model is deterministic, time-periodic (hence not ergodic) and time-reversible. It is
ideal, therefore, for illustrating how irreversibility emerges in the macroscopic limit.
In previous works on this model, e.g., it has been established that, on average, the fractional difference between the number of black and white balls tends to zero for times that are large, but smaller than the period (which depends on $N$), and that the variance of this observable also vanishes in this limit.

This, however, does not completely establish the validity of Boltzmann's full scenario, even for
this simple model system.
Indeed, it only establishes the behaviour of ensemble averaged quantities and does not prove that this irreversible
approach to equilibrium emerges for any individual system that starts from one of an overwhelmingly numerous set
of typical initial states. Nor does it show that, once so-equilibrated, such a system will almost certainly remain
in equilibrium for times that are, like $N$ itself, exceedingly large. As in our treatment of the expanding gas, we focus in the present paper on this latter situation, and in particular we prove in Theorem~\ref{thm:KacRing} that it does. 
 For the Kac ring, we improve upon previously obtained algebraic bounds for the deviation from unity of relevant equilibration probabilities, by providing (sub)-exponential bounds for them of the form  $\exp (-\gamma N^{1-\alpha})$.

\begin{figure}
\centering
\captionsetup{width=0.85\textwidth}
\includegraphics[height =1.5in]{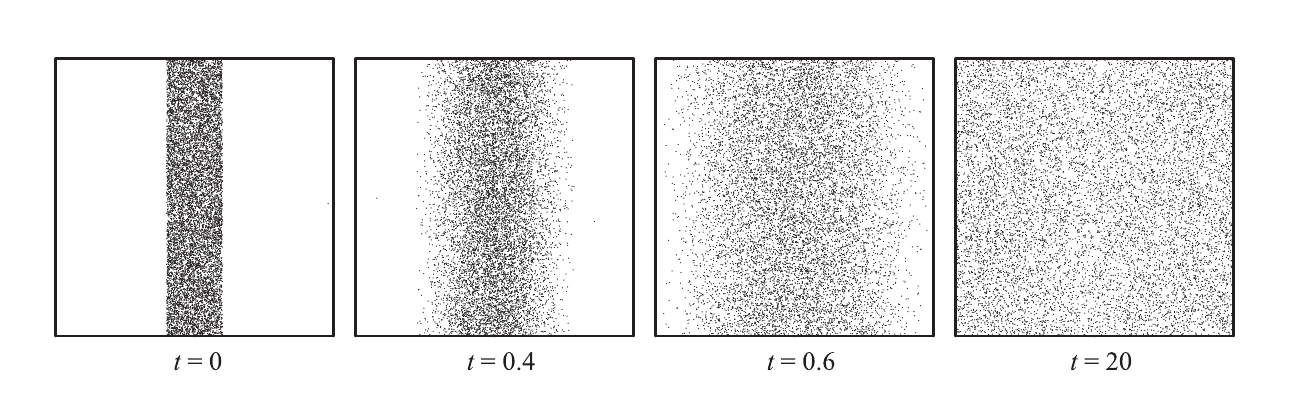}
\caption{Numerical simulation of the free expansion on the 2-torus of a non-interacting gas of $10^4$ particles having a thermal distribution of momenta with mean thermal speed equal to unity,  initially confined in the horizontal direction to the region $0.4 <x < 0.6$, at the sequence of times indicated.}
\label{fig:expand}
\end{figure}

The paper is organized as follows. In Section~\ref{s:freeexpansion} we
describe our model of the expanding gas and give an informal mathematical
statement of our main result.
In Section~\ref{s:meanvariance} we provide a first set of
results that show that an ensemble of non-interacting freely-expanding gases will,
on average, develop a flat density profile.
We then present in Section~\ref{s:typical} a more precise definition of those microstates
that exhibit ``typical'' behaviour, after which we state and prove our main result on the expanding
gas. Section~\ref{s:examples} contains numerical examples illustrating the main theorem along with a discussion of its natural generalizations to other geometries.
Section~\ref{s:Kacring} contains entirely analogous results on the Kac ring model.
Further discussion and a summary is provided in Section~\ref{s:discussion}.

\noindent\textbf{Acknowledgements} S.D.B. is supported by the Labex CEMPI (ANR-11-LABX-0007-01) and by the Nord-Pas de Calais Regional Council and the Fonds Europ\'een de D\'eveloppement \'Economique R\'egional (grant CPER Photonics for Society). P.E.P. thanks the University of Lille and the Labex CEMPI, where part of this work was performed, for their hospitality.  The authors thank H. Spohn for bringing the work of J.~Beck to their attention, and the latter for communicating his recent unpublished work to them.

\section{A freely expanding gas}\label{s:freeexpansion}
Consider particles moving on a $d$-dimensional torus $\T^d$, identified with the
cube $[0,1]^d$ in what follows. The phase space of one particle
is $\T^d\times \R^d$ and that of $N$ particles is
$$
\Gamma_N=\left(\T^d\times\R^d\right)^N.
$$
A point in the phase space of the $N$-particle system will be written
$(X,P)=(x_{i},p_{i})_{i=1,\dots N}$. The dynamics is that of uniform
rectilinear motion: the particles are not subject to any force and do not
interact.

Initial conditions are constructed by
first choosing an a priori initial probability measure $\mu$ on $\Gamma_{1}$, and
then drawing, for each of the $N$ particles, initial phase space points
$(x_{i},p_{i})$ independently from this probability measure. The only randomness in
the model is in this choice of the initial condition $\left(  X,P\right)  $
drawn from the product measure induced by $\mu$ on $\Gamma_{N}$.

The focus of our analysis will be on the following coarse-grained observables for the system. Given any measurable subset $I$ of the configuration space $\T^d$, we consider the fraction
\begin{equation}
f_{I}(X,P,t)=\frac{1}{N}\sum_{i=1}^{N}\chi_{I}(x_{i}+p_{i}t),
\label{eq:Ifraction}%
\end{equation}
of particles inside $I$, which is expressed above in terms of the periodicized characteristic function $\chi_{I}$ of
$I$; the latter vanishes everywhere except at points $y$ such that $\{y\}\in I$,
where it is equal to unity,  and where  $\{y\}$ denotes the fractional part
of $y$ taken componentwise.
Letting $\{I_\alpha  |  \alpha = 1,\dots, L\}$
be a partition of the configuration space,
we associate a coarse-grained time-dependent density profile $\rho_\alpha(t) = Nf_{I_\alpha} (X,P,t))/ |I_{\alpha}|$
to each initial state $(X,P)\in\Gamma_N$.

Our main result is as follows: given an initial probability measure $\mu$, satisfying suitable conditions spelled out in the next section, we show that for an overwhelming majority of initial conditions, there exists a relatively short time $t_0$ depending only on $\mu$, and a comparatively long time $t_N$ that grows exponentially with  $N$, such that, for all $\alpha$, the fraction $f_{I_\alpha}$ of particles that finds itself in $I_\alpha$ at times $t \in [t_0, t_N]$ is close to the measure $|I_\alpha|$ of that subset.
Such initial conditions, the probability of which we will show to differ from unity by an amount that is exponentially small in $N$, we will call ``typical''.

In other words, for typical initial conditions, the $N$ particles are, after a relatively brief equilibration time, uniformly distributed on the torus after which they will then be observed to be so for a very long time.
In this sense, the gas exhibits an approach to equilibrium.

We only consider out-of-equilibrium initial states in which the particles are uncorrelated. In general out-of-equilibrium initial states, one expects there to be correlations between the particles. For some such states, equilibration will not occur, for example if all velocities are identical and aligned with an axis of the torus. We have not dealt with the more general question of identifying the general class of initial states for which equilibration does indeed occur.

\section{Mean values}\label{s:meanvariance}
We start our analysis by studying the mean and the variance of the coarse-grained
observables $f_I$ as a function of $I,N$, and $t$ in order to establish, at the outset, the asymptotic values that these observables are expected to approach at sufficiently long times. The main result of this Section is stated in Proposition~\ref{lem:L1L2conv}.

We will obtain results for two classes of the initial probability measure $\mu$. We first consider the case where
\begin{equation}\label{eq:rhoconditions1}
\rd \mu(x,p)=\rho(x,p)\rd x\rd p,\quad \rho\in L^1(\R^d, L^2(\T^d)),
\end{equation}
so that, in particular, $\mu$ is absolutely continuous with respect to Lebesgue measure. In addition, we need the density $\rho$ to be
sufficiently smooth in the momentum variables and to satisfy the following condition:
\begin{eqnarray}\label{eq:rhoconditions}
\exists r > d/2, \forall k\in \N, 1\leq k\leq 2r,\quad
\int_{\Gamma_1}\mid \partial_p^{k}\rho(x,p)\mid\rd x\rd p<+\infty.
\end{eqnarray}
This allows, e.g., for a product distribution $\rho(x,p) = \chi_{I_0}(x)\rho_{{\mathrm p}}(p)$ that
has all particles initially confined to a measurable subset $I_0$ of the torus,
but which has a sufficiently smooth (e.g., thermal) momentum distribution. Obviously,
this is a generalization of the situation that takes place in an already thermalized
gas when the volume to which it is constrained is suddenly expanded.

To proceed, we first need the following simple estimate, which is the only place where the one-particle dynamics plays a role. We will write
\begin{equation}\label{eq:chiIt}
\chi_{I,t}(x,p)=\chi_I(x+pt),\quad \forall (x,p)\in \Gamma_1.
\end{equation}
\begin{lemma} \label{lem:oneparticle}
Suppose $\rho$ is as above. Then there exists a constant $C_r$ such that, for all measurable subsets $I$ of $\T^d$ and for all $t\in\R_*$,
\begin{equation}\label{eq:expestimate}
\mid \E(\chi_{I,t})-|I|\mid\leq C_\mu t^{-2r}.
\end{equation}
where $C_\mu=C_r\|\partial_p^{2r}\rho\|_1$.
\end{lemma}
\noindent Here
\begin{equation*}
\E (\chi_{I,t})=\int_{\Gamma_1} \chi_I(x+pt) \rd \mu,
\end{equation*}
and so, in particular,
\begin{equation}\label{eq:sortamixing}
\lim_{t\to+\infty} \int_{\Gamma_1} \chi_I(x+pt) \rd \mu = |I| .
\end{equation}

To put this result in perspective, recall that, with respect to Lebesgue measure, for almost every $p\in\R^d$, the flow $x\to x+pt$ on the $d$-torus is ergodic. Hence, for  almost every $(x,p)\in \Gamma_1$,
\begin{equation}\label{eq:ergodic1particle}
\lim_{T\to+\infty}\frac1T\int_0^T \chi_I(x+pt)\rd t=|I|.
\end{equation}
This statement means that, for almost every initial
condition, a single free particle  spends approximately a
fraction $|I|$ of its time inside $I.$ It is therefore true pointwise in $(x,p),$ but averaged
in time. Equation~\eqref{eq:sortamixing}, on the other hand, holds pointwise in $t$, but averaged in $(x,p)$.

Although the free dynamics of the single particle is not mixing in
$\Gamma_{1}$, the limit in~\eqref{eq:sortamixing} has a ``mixing'' flavour to
it. Indeed, the map $I\rightarrow\mathbb{E}(\chi_{I,t})\in\lbrack0,1]$ defines
for each $t\in\mathbb{R}$ a measure on $\T^{d}$ and
Lemma~\ref{lem:oneparticle} shows that these measures converge to Lebesgue
measure as $t\rightarrow+\infty$. Notice that the estimate in~\eqref
{eq:expestimate} deteriorates when the density becomes more concentrated and
breaks down for any fixed initial condition $\rho(x,p)=\delta
(x-x_{0},p-p_{0})$, as does~\eqref{eq:sortamixing}. We will come back to this point below, when dealing with the second class of measures for which our results hold (See~\eqref{eq:productmeasure}-\eqref{eq:deltacondition}).

\begin{proof}
Let $\rho$ be as above, and define, for $a\in (2\pi\Z)^d, b\in\R^d$,
$$
\hat\rho(a,b)=\int_{\Gamma_1} \exp(-i(ax+bp))\rho(x,p)\rd x\rd p.
$$
Then, for all non-zero $b\in\R^d$,
$
|\hat\rho(a,b)|\leq \|b\|^{-2r}\|\partial_p^{2r}\rho\|_1.
$
Now, defining, $\forall \ell\in\Z^d$,
$$
\tilde\chi_\ell=\int_{\T^d}\exp(i2\pi\ell x) \chi_I(x)\rd x,
$$
we have
\begin{eqnarray}
\E(\chi_{I,t})&=&\int_{\Gamma_1}\chi_I(x+pt)\rho(x,p)\rd x\rd p=\sum_{\ell\in\Z^d}\tilde \chi_\ell \hat\rho(2\pi\ell, 2\pi \ell t)\nonumber\\
&=&\tilde\chi_0 + \sum_{\ell\not=0} \tilde \chi_\ell \hat\rho(2\pi\ell, 2\pi \ell t).\label{eq:plancherel}
\end{eqnarray}
Hence, using the fact that $|\tilde \chi_\ell|\leq 1$, \eqref{eq:expestimate} follows.
\end{proof}
We now turn to the $N$-particle problem. Let
\begin{equation}\label{eq:fIt}
f_{I,t}(X,P)=f_I(X,P,t),\quad \forall (X,P)\in\Gamma_N.
\end{equation}
 We will denote by $\E_N$ the expected value with respect to the product probability measure induced by $\mu$ on $\Gamma_N$. A first formulation of the approach to a uniform spatial density for the gas is then as follows.\\
\begin{proposition} \label{lem:L1L2conv} Let  $\mu$ satisfy~\eqref{eq:expestimate}. Then, for all measurable subsets $I$ of the one-particle configuration space $\T^d$, we have
\begin{eqnarray}
\forall N\in\N_*,\ \lim_{t\to+\infty}\E_N(f_{I,t}-|I|)&=&0,\label{eq:convL1}\label{eq:gasexpansionmean}\\
\lim_{N\to+\infty,t\to+\infty}\E_N\left((f_{I,t}-|I|)^2\right)&=&0.\label{eq:convL2}
\end{eqnarray}
The limits are uniform in $I$.
\end{proposition}
\begin{proof}
Since the initial conditions $(x_i, p_i)_{i=1\dots N}$ are drawn independently from the single-particle probability $\mu$, the course-grained observable $f_I(X,P,t)$ is a sum of \emph{i.i.d.} random variables. One then trivially obtains, for all $N$ and $t$
\begin{eqnarray}
\E_N\left(f_{I,t}-\E(\chi_{I,t})\right)&=&0\nonumber\\
\E_N\left((f_{I,t}-\E(\chi_{I,t}))^2\right)&=&\frac1N\V(\chi_{I,t}),\label{eq:variance}
\end{eqnarray}
where
\begin{equation*}
\V(\chi_{I,t})=\int_{\Gamma_1} (\chi_I(x+pt) - \E(\chi_{I,t}))^{2} \rd\mu,
\end{equation*}
so that, using~\eqref{eq:expestimate}, $\forall N\in\N_*, \forall t>0$,
\begin{eqnarray}
\mid\E_N(f_{I,t}-|I|)\mid&=&\mid \E(\chi_{I,t})-|I|\mid\nonumber\\
&\leq& C_\mu t^{-2r},\label{eq:control1}\\
\E_N\left((f_{I,t}-|I|)^2\right)&=&\frac{1}{N}\V(\chi_{I,t})+\left(\E(\chi_{I,t}-|I|)\right)^2\nonumber\\
&\leq& \frac1N + \left( C_\mu t^{-2r}\right)^2,\label{eq:control2}
\end{eqnarray}
where we have used the fact that $\V(\chi_{I,t})\leq 1$ for all $t$ and $I$. This implies~\eqref{eq:convL1} and~\eqref{eq:convL2}.
\end{proof}
We now present a second class of measures $\mu$ on $\Gamma_1$ for which~\eqref{eq:expestimate} and hence Proposition~\ref{lem:L1L2conv} can be proven to hold.  Le $\mu$ be a product measure
\begin{equation}\label{eq:productmeasure}
\rd \mu(x,p)=\rd \nu(x)\rho_{\mathrm p}(p)\rd p ,
\end{equation}
where $\nu$ is an arbitrary probability measure $\rd \nu$ on $\T^d$ and where $\rho_{\mathrm{p}}$ is a one-particle momentum distribution. Introducing
$$
\hat\rho_{\mathrm{p}}(b)=\int_{\R^d}\exp(-ibp)\rho_{\mathrm{p}}(p)\rd p,
$$
we will suppose there exists constants $C$, $r>d/2$, and $s>d$ such that
\begin{equation}\label{eq:rhopcondition}
\hat\rho_{\mathrm{p}}(b)\leq C\|b\|^{-2r},\qquad |\rho_{\mathrm{p}}(p)|\leq C(1+\|p\|)^{-s}.
 \end{equation}
 This, e.g.,
allows the application of our results to situations such that
\begin{equation}\label{eq:deltacondition}
\rd \nu = \delta(x-a_0)\rd x ,
\end{equation}
in which all particles are initially located at the same position $a_0$ and in which, therefore,
the only randomness that appears is in the choice of the initial momenta.

We now show how to establish~\eqref{eq:expestimate} for this class of measures.
The constant $C_\mu$ will depend on $C,r$, and $s$.
For that purpose, it is in turn sufficient to establish the analog of~\eqref{eq:plancherel}, i.e.,
\begin{equation}\label{eq:planchereltris}
\E(\chi_{I,t})=\sum_{\ell\in\Z^d}\tilde\chi_\ell\tilde\nu(\ell)\hat\rho_{\mathrm{p}}(2\pi\ell t),
\end{equation}
where
$$
\tilde\nu(\ell) = \int_{\T^d}\exp(i2\pi\ell x) \rd \nu.
$$

To that end, we consider a sequence $\chi_I^{(n)}$ of $C^{\infty}$ functions that approximate $\chi_I$ in  the $L^2$-norm
\begin{equation}\label{eq:smoothapprox}
\lim_{n\to+\infty} \|\chi_I^{(n)}-\chi_I\|_2=0.
\end{equation}
Writing $\tilde\chi_\ell^{(n)}$ for their Fourier coefficients, it follows that $\lim_{n\to+\infty}\tilde\chi_\ell^{(n)}=\tilde\chi_{\ell}$
and the smoothness of the $\chi_I^{(n)}$ implies that the Fourier coefficients
decay fast in $\ell$. Then, for all $x\in\T^d$,
$$
\chi_I^{(n)}(x)=\sum_{\ell} \tilde\chi_{\ell}^{(n)}\exp(-i2\pi \ell x).
$$
The convergence of this series is uniform in $x$ because of the fast decay of the Fourier coefficients. Then, for all $t\in\R$,
\begin{eqnarray}
\E(\chi_{I,t})&=&\int_{\T^d}\left(\int_{\R^d} \chi_I(x+pt))\rho_{\mathrm{p}}(p) \rd p\right)\rd\nu(x)\nonumber\\
&=&\int_{\T^d}\left(\int_{\R^d} \left(\chi_I(x+pt)-\chi_I^{(n)}(x+pt)\right)\rho_{\mathrm{p}}(p) \rd p\right)\rd\nu(x)\nonumber\nonumber\\
&\ &\qquad+\int_{\T^d}\left(\int_{\R^d} \chi_I^{(n)}(x+pt))\rho_{\mathrm{p}}(p) \rd p\right)\rd\nu(x).\nonumber\\\label{eq:chiapprochebis}
\end{eqnarray}
Let $C_\ell$, $\ell\in\Z^d$ be the family of cubes centered at $\ell\in\Z^d$ and of linear size $1$. Then we can write, for all $x\in\T^d$,
\begin{eqnarray*}
|\int_{\R^d} \left(\chi_I(x+pt)-\chi_I^{(n)}(x+pt)\right)\rho_{\mathrm{p}}(p)\rd p|=\qquad\qquad\\
\qquad\qquad\frac1{t^d}|\sum_\ell\int_{C_\ell}\left(\chi_I(x+y)-\chi_I^{(n)}(x+y)\right)\rho_{\mathrm{p}}(y/t)\rd y|\\
\leq\frac1{t^d}\|\chi_I-\chi_I^{(n)}\|_2\sum_\ell\left(\int_{C_\ell}\rho_{\mathrm{p}}^2(y/t)\rd y\right)^{1/2}\qquad\\
\leq \tilde{C}_{s,t}\|\chi_I-\chi_I^{(n)}\|_2.\qquad\qquad\qquad\qquad\qquad\quad
\end{eqnarray*}
Inserting this into~\eqref{eq:chiapprochebis}, one finds, for all $t\in\R$,
$$
\E(\chi_{I,t})=\lim_{n\to+\infty}\int_{\T^d}\left(\int_{\R^d} \chi_I^{(n)}(x+pt)\rho_{\mathrm{p}}(p) \rd p\right)\rd\nu(x)
$$
Hence
$$
\E(\chi_{I,t})=\lim_{n\to+\infty}\sum_{\ell\in\Z^d}\tilde\chi_\ell^{(n)}\tilde\nu(\ell)\hat\rho_{\mathrm{p}}(2\pi\ell t).
$$
Using condition~\eqref{eq:rhopcondition} and bounded convergence, one finds finally~\eqref{eq:planchereltris} and hence~\eqref{eq:expestimate}.

\begin{figure}[t]
\captionsetup{width=0.85\textwidth}
\includegraphics[height=3in]{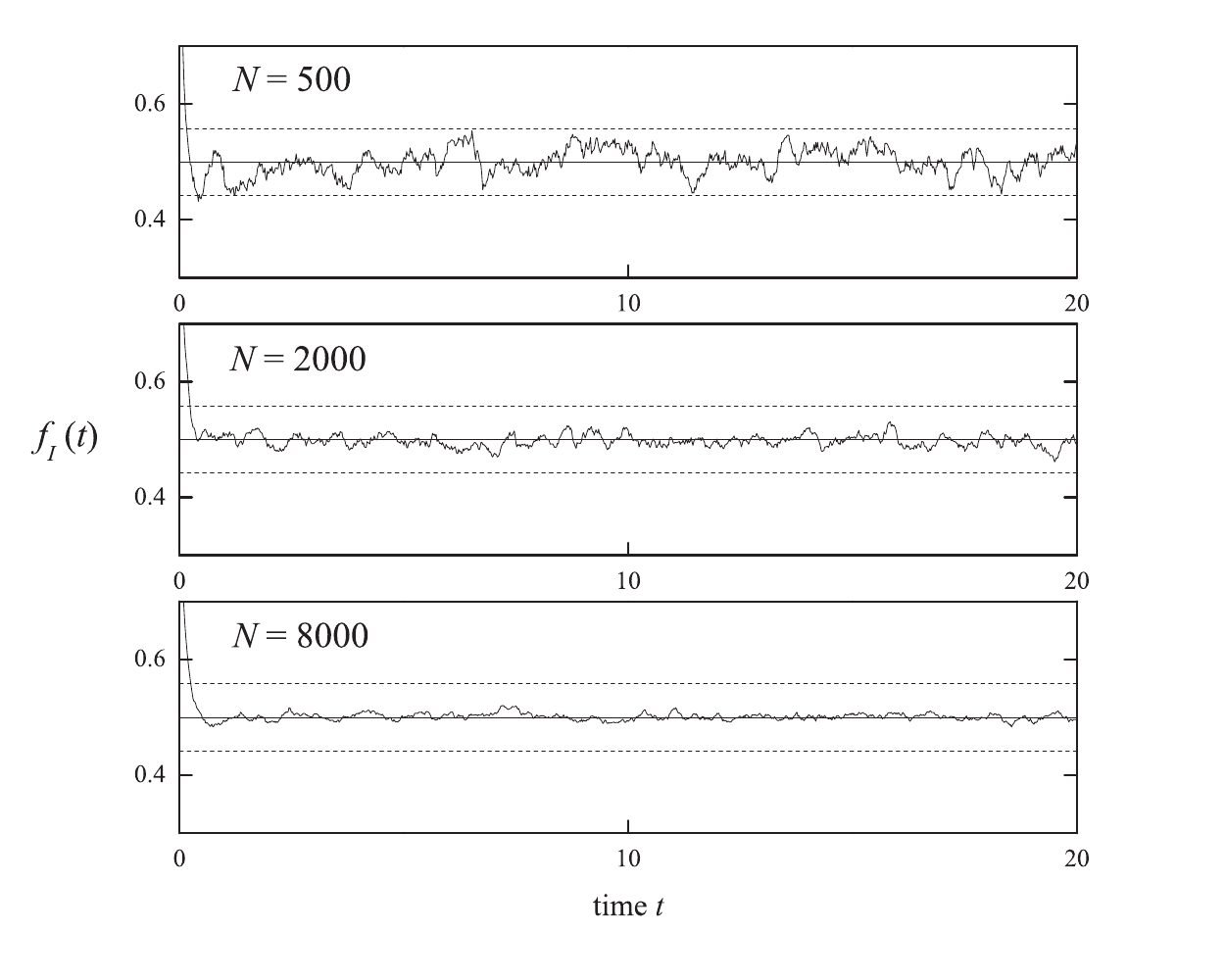}
\centering
\caption{Evolution over time of the fraction $f_{I}(X,P,t)$ of $N$ particles (with values of $N$ as
indicated) located at time $t$ in the interval $I = [0,0.5]$ for a single initial
condition $(X,P)$ of the gas. To choose this initial condition, each of
the $N$ particle positions was independently chosen  from a uniform distribution in $[0, 0.5]$
and  each of their momenta was independently chosen from a thermal distribution with mean
thermal speed equal to unity.}
\label{fig:fluctuations1}
\end{figure}

The model we study here is a variation of one that was introduced in~\cite{Fr58},
where a result equivalent to~\eqref{eq:gasexpansionmean} is proven. In~\cite{Be10} (Theorem~3) a result similar to~\eqref{eq:convL2} is proven on a closely related model.
 Equations~\eqref{eq:gasexpansionmean}-\eqref{eq:convL2} mean
that for large enough $N$ and $t$, and for all $I\subset\T^{d}$, the
coarse-grained random variable $f_{I}$ has a mean close to $|I|$ and a very small
variance. This implies that for any partition $\{I_\alpha | \alpha\in\{1,\dots, L\}\} $
of the configuration space, the coarse-grained time-dependent density
profile $\rho_\alpha(t) = Nf_{I_\alpha} (X,P,t))/ |I_{\alpha}|$ is, \emph{on average}, and for sufficiently large $N$ and $t$, close to being spatially homogeneous.
This provides a weak sense of the approach to spatial homogeneity. Indeed, the idea is that when a random
variable has a very small variance it is ``almost constant''. In that sense, one may claim that $f_{I_\alpha}(X,P,t)$ is ``typically'' equal to $|I_\alpha|$, provided $N$ and $t$ are large enough.
Note however
that~\eqref{eq:convL1}-\eqref{eq:convL2} do not as such
give any information about the approach of the coarse-grained observables $f_{I_\alpha}(X,P,t)$ to $|I_\alpha|$ for any
particular initial condition $(X,P)\in\Gamma_{N}$, i.e., about the evolution of individual macroscopic
systems.
Figure~\ref{fig:fluctuations1}, on the other hand, illustrates for single trajectories of individual $N$-particle systems, how typical fluctuations in $f_I$ about its equilibrium value decrease with increasing particle number $N$, but do not decay to zero as a function of time. In this figure, individual trajectories, after a relatively brief initial equilibration time that does not appear to depend on $N$, display values of $f_I$ that fluctuate about the expected mean value of  $0.5$.

Note that, according to~\eqref{eq:gasexpansionmean}, the average of $f_I - |I|$ over many such trajectories will approach zero with increasing $t$. On the other hand, as is evident from the data in Fig.~2, and as is confirmed below in~\eqref{eq:control2}, the variance of that quantity at fixed $N$ approaches, for large times, a finite value, which can be estimated from the size of the temporal fluctuations appearing in the figure.  Numerical results evaluated at times very much longer than the relatively short initial equilibration period, i.e., at times much longer those shown in Fig.~2, reveal fluctuations in $f_I$ statistically equivalent to those appearing in that figure.

For these simulations, which were performed with one-dimensional systems with particle numbers $N$ as indicated, the individual initial conditions for each trajectory displayed were chosen so that all $N$ particles were randomly located in the interval $I=[0,0.5]$, with momenta independently chosen from a thermal distribution with average value equal to unity. The value of $f_I$ displayed in these figures  is then
the fraction of particles in that same interval as a function of time.

Note that the estimates~\eqref{eq:expestimate} and~\eqref{eq:variance} can already be used to obtain some information on initial trajectories, via the Markov inequality, in the spirit of a quantitative weak law of large numbers. To see this, let us introduce for all measurable $I$ and $\epsilon$, and for all $t>0$, and $N\in\N_*$ the set
\begin{equation}\label{eq:goodset}
\Omega_I(\epsilon, N, t)=\{(X, P)\in\Gamma_N\mid |f_I(X,P,t)-|I|\mid\leq \epsilon\}
\end{equation}
of initial conditions for which the fraction of particles that are inside $I$ at time $t$ is close to $|I|$ when $\epsilon$ is small. For any subset $\Omega$ of the $N$ particle phase space, we will denote by $P(\Omega)$ the probability that an initial condition $(X,P)$, drawn from the product probability measure induced by $\mu$, lies in $\Omega$. One then obtains readily that, for large enough $t$, depending on $\epsilon$ and $\mu$, but not on $N$
$$
P(\Omega_I(\epsilon, N, t))\geq 1-\frac{4}{\epsilon^2N}.
$$
However, the control in $N^{-1}$ on this probability is unsatisfactory.
A more detailed mathematical analysis of the main features of $f_I$ for individual trajectories, as observed in Fig.~\ref{fig:fluctuations1}, will be provided in the next section, where an exponential estimate on the above probability will be obtained.

\section{Typical behaviour}\label{s:typical}
To fully implement Boltzmann's scenario we now show that for an overwhelming majority of initial conditions drawn from
the product probability measure generated by $\mu$, there exists an exponentially long sequence of times at which the fraction $f_{I_\alpha}$ of particles in $I_\alpha$ is arbitrarily close to the measure $|I_\alpha|$ of that subset, for all $\alpha =1, \dots, L$. Here, the $I_\alpha$ form a partition of the torus, as described in Section~\ref{s:freeexpansion}.

To state our result, we introduce a sequence of times $t_k=t_0+k\Delta t$, $k=1,\dots , K$, for some $\Delta t$ that  is macroscopically small enough that a series of observations of the system at these times could be considered quasi-continuous. The total duration over which the system is observed is therefore $\Delta T=t_K-t_0=K\Delta t$. Define, then, the set
$$
\Omega_{I, K} =  \cap_{k=1}^{K}\ \Omega_I(\epsilon, N, t_k),
$$
of ``good'' initial conditions in which the fraction $f_{I}$ of particles in $I$ is close to $|I|$ at each instant of time in the sequence. Note that we have suppressed the $\epsilon$ and $N$ dependence of $\Omega_{I, K}$.  Let $P(\Omega_{I, K})$
denote the probability that an initial condition, chosen as described, lies in $\Omega_{I, K}$.  We will show, see~\eqref{eq:BoltzmannScenario}, that for values of $K=K_N$ that are exponentially large in $N$, the
probability $P(\Omega_{I, K_N})$ differs from unity by an amount that is exponentially small in $N$.

\begin{theorem}\label{thm:gasexpansion}
Let $\mu$ be a probability measure on $\Gamma_1$  and suppose there exists $C_\mu>0$,
and $r>0$ so that
\begin{equation}\label{eq:expestimategeneral}
|\E(f_{I,t})-|I||\leq C_\mu t^{-2r}.
\end{equation}
Let $$
t_0=\epsilon^{-\frac1{2r}} D_\mu,\quad\mathrm{where}\quad D_\mu=\left(2C_\mu\right)^{\frac1{2r}}.
$$
Then, for all $K\in\N_*$,
$$
P(\Omega_{I,K})\geq 1-2K\exp(-\frac12\epsilon^2N).
$$
In particular, if $K_N=\frac12\exp(\frac{\epsilon^2}{4}N)$, then
\begin{equation}\label{eq:BoltzmannScenario}
P(\Omega_{I, K_N})\geq 1-\exp(-\frac14\epsilon^2N).
\end{equation}
\end{theorem}
Note that it follows from~\eqref{eq:BoltzmannScenario} that
\begin{equation}\label{eq:BoltzmannScenariobis}
P(\cap_{\alpha=1}^L\Omega_{I_\alpha, K_N})\geq 1-L\exp(-\frac14\epsilon^2N).
\end{equation}
In Section~\ref{s:meanvariance} we identified two classes of measure for which condition~\eqref{eq:expestimategeneral} is satisfied.
For the proof of Theorem~\ref{thm:gasexpansion}, we will utilize a particular case of Theorem~2 proven in~\cite{Ho63}:
\begin{theorem}\label{thm:largedev} Let $X_i$, $i=1\dots N$ be a family of independent identically distributed random variables
with $0\leq X_i\leq 1$ and $\E(X_i)=m$. Let $\epsilon>0$. Then
$$
P(\mid \frac1N\sum_{i=1}^N X_i-m\mid\geq \epsilon)\leq 2\exp(-2\epsilon^2N).
$$
\end{theorem}
\noindent With this result, then:\\
{\em Proof of Theorem~\ref{thm:gasexpansion}.} It follows from~\eqref{eq:expestimategeneral} that,
 for all $t>t_0$,
 $$
\mid \E(\chi_{I,t})-|I|\mid\leq \epsilon/2.
$$
It follows that, for all $t>t_0$, $(X,P)\in\Omega_I^c(\epsilon, N, t)$ implies that
$$
\mid f_I(X,P,t)-\E(\chi_{I,t})\mid\geq \frac\epsilon2.
$$
Hence, by Theorem~\ref{thm:largedev},
\begin{equation}\label{eq:estimatefixed_t}
P(\Omega_I^c(\epsilon, N, t))\leq P\left(\mid f_I(X,P,t)-\E(\chi_{I,t})\mid\geq \frac\epsilon2\right)\leq 2\exp(-\frac12\epsilon^2N),
\end{equation}
and hence
\begin{equation}\label{eq:P(good)}
P(\cap_{k=1}^{K_N}\Omega_I(\epsilon, N, t))\geq 1- 2K_N\exp(-\frac12\epsilon^2N),
\end{equation}
which proves the result.
{\hfill\ensuremath{\square}}\\
Note that the equilibration time $t_0$ is independent of $N$, as observed also in the numerics presented in Section~\ref{s:meanvariance}. It is directly linked to the estimate in~\eqref{eq:expestimate} and in no clear way related to the ergodic convergence time in~\eqref{eq:ergodic1particle}, which is strongly $p$-dependent.

After we completed an initial version of this paper, Prof.~J.~Beck kindly communicated to us recent unpublished work~\cite{Be17} in which he proves, using detailed dynamical information, combinatorial estimates, and Fourier analysis,  a result similar to Theorem~\ref{thm:gasexpansion}, for initial conditions
$$
\mu_N(x,p)=C_{N}\delta(x_1-a_1)\dots \delta(x_N-a_N)\exp(-\sum_i \|p_i\|^2),
$$
in which the particle positions are initially non-random, as in~\eqref{eq:deltacondition}, and their velocities are specifically drawn from a Maxwell distribution.

\section{Examples and generalizations }\label{s:examples}
\subsection{Examples}
The plot in Fig.~\ref{fig:Fig3} illustrates in one dimension the exponential scaling with $N$ and the linear scaling with $K$ of the fluctuations implied by~\eqref{eq:P(good)}.
It shows the fraction $P_\epsilon(N)$ (normalized by $K$) of $M = 3\times 10^7$ initial conditions, randomly chosen as described below, that exhibited a deviation in the
quantity $|f_I - |I||$, for $I = [0,0.5]$, of magnitude greater than or equal to $\epsilon = 0.04$, at least once over a sequence of fixed times $t_k = k\Delta t$, with $t_0 = \Delta t = 10$ and $k = 1,\dots ,K$, for values of $K=1,5,10,15,20$, and $25$.  For each initial condition, particle positions were independently and uniformly chosen from the same interval $I$, and momenta were independently chosen from
a thermal distribution with a mean thermal velocity of unity. The six closely bunched dotted curves include the fraction of such initial conditions for each value of $K$ indicated, excluding those (at large $N$ and small $K$) for which no such fluctuations were obtained during the $M = 3\times 10^7$ gas histories, over the time scale investigated. Open circles indicate an average over $K$. The dot-dashed line near the top of the figure represents the function $P_\epsilon(N) /K = 2\exp(-\frac12\epsilon^2N)$ that, according to~\eqref{eq:P(good)}, serves as an upper bound for the normalized probability.  The solid line that follows the numerical data for large $N$ is a fit to the data for $N\geq 2000$ of the form $P_{\epsilon}(N)/K = a \exp(-b\epsilon^2 N)$, with $a=0.14$ and $b=2.38$, showing that the exponential bound implied by~\eqref{eq:P(good)} is a conservative one.
\begin{figure}
\includegraphics[height =3in]{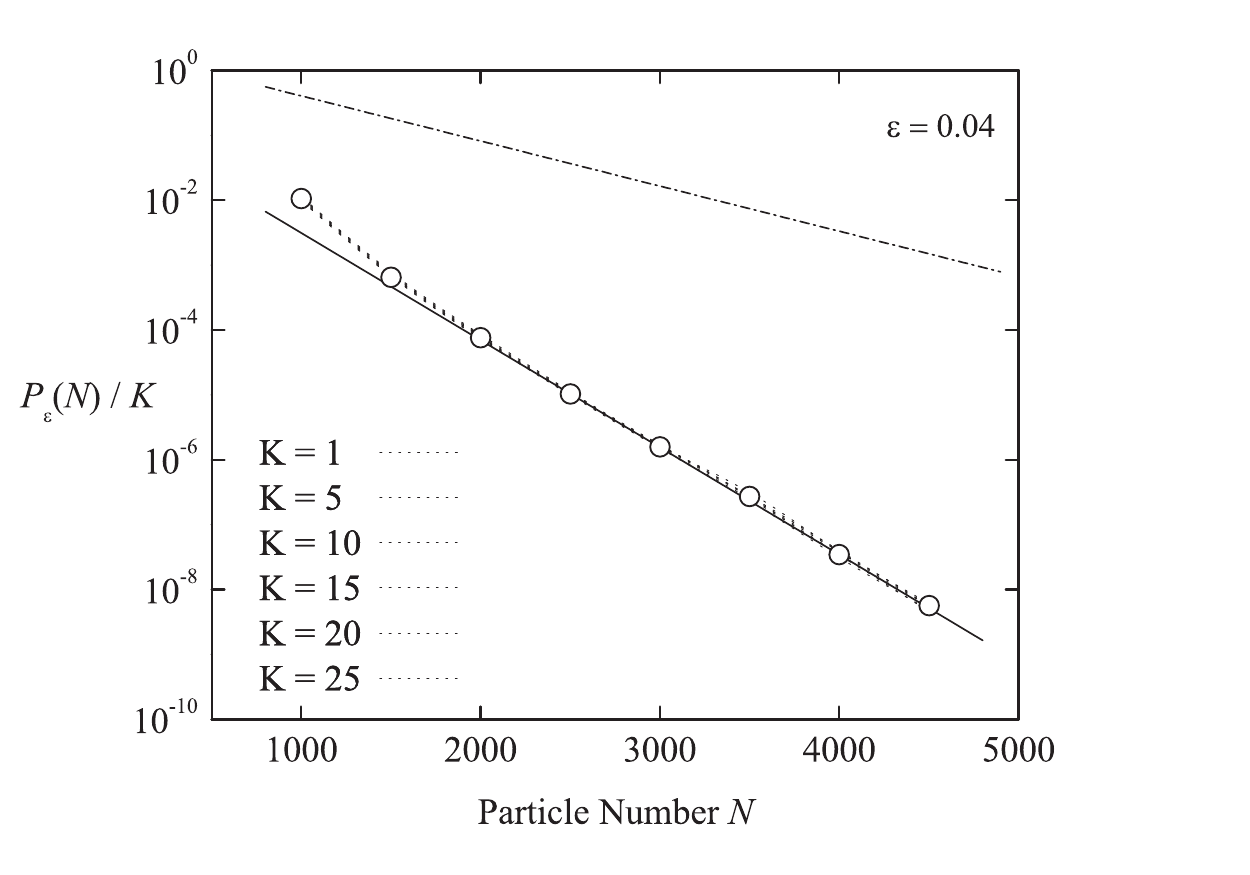}
\captionsetup{width=0.85\textwidth}
\centering
\caption{Scaling with $N$ of the probability $P_\epsilon(N)$ of a fluctuation from equilibrium as described in the text.}
\label{fig:Fig3}
\end{figure}

In fact, the proof of~\eqref{eq:P(good)} can easily be adapted to show that, for all $0<\eta<1$, for $t_0(\eta)=\epsilon^{-\frac1{2r}}\left(\frac1{\eta}C_\mu\right)^{\frac1{2r}}$, and for all $K\in\N_*$, one has
\begin{equation}\label{eq:P(goodbis)}
P(\cap_{k=1}^{K}\Omega_I(\epsilon, N, t))\geq 1- 2K\exp(-2(1-\eta)^2\epsilon^2N).
\end{equation}
In the proof above, we used $\eta=\frac12$. If, on the other hand, one takes $\eta$ small, one gets a bound close to $\exp(-2\epsilon^2N)$, much closer to what is observed numerically. For example, with $\epsilon=0.04$ and $\eta=10^{-2}$, one finds $t_0(\eta)\simeq 20$ and $2(1-\eta)^2\simeq 1.96$. Constraints on computing power make it unfeasible to do simulations with values of $N$ or $K$ much larger than those used here to illustrate the theorem. For example, with $N=8 000$ and $\epsilon=0.04$, the probability for finding a single fluctuation is of order $\exp(-2(0.04)^2N)\simeq 7.6\times 10^{-12}$. Identification of such a fluctuation would require the generation of  $10^{12}$ or more different initial conditions, exceeding realistic computational capabilities.

A realistic numerical simulation showing the applicability of Theorem~\ref{thm:gasexpansion} to actual macroscopic systems being, therefore, out of the question, we instead now use the theorem to obtain simple estimates for a more-or-less realistic hypothetical macroscopic system.

Thus, we consider in three dimensions a gas cell occupying a cubic domain $D$ of total volume $|D| =1$~cm$^{3}$
that is filled with an ideal gas at standard temperature and pressure. The equilibrium particle number density of
such a gas is well-known to be of the order $n_{0}\sim 3\times 10^{19}$/cm$^{3}$.
Thus, ignoring fluctuations for the moment, let us assume the actual number of gas particles
in the cell to be
\begin{equation*}
N=n_{0}\left\vert D \right\vert =3\times 10^{19}.
\end{equation*}

Consider, next, a cubic sub-region $D _{\alpha }$, lying entirely within $D$,
and having a linear dimension one-tenth that of the total domain $D$.
It occupies, therefore, a volume $\left\vert D _{\alpha }\right\vert =$ $1~$%
mm$^{3}$, i.e., a fraction $\left\vert I_{\alpha
}\right\vert =\left\vert D _{\alpha }\right\vert /\left\vert D
\right\vert =10^{-3}$ of the total volume of the gas cell.

In thermal equilibrium at temperature $T$, the average number
\begin{equation*}
\langle N_{\alpha }\rangle =n_{0}\left\vert D _{\alpha }\right\vert
=3\times 10^{16}
\end{equation*}
of particles in such a sub-region corresponds to a fraction $f_{\alpha }=\langle
N_{\alpha }\rangle /N=|I_{\alpha }|=10^{-3}$ of all the particles in the
system, while the average pressure in this sub-region, being an intensive quantity,
is the same as in the bulk, i.e.,
\begin{equation*}
p_{eq}=\frac{N}{|D|}kT=\frac{\langle N_{\alpha }\rangle } {|D_{\alpha }| }kT=n_{0}kT
\end{equation*}%
In what follows, we will denote by
\begin{equation*}
p_{\alpha }\left( t\right) =\frac{N_{\alpha }(t)}{\left\vert D _{\alpha
}\right\vert }kT=n_{\alpha }\left( t\right) kT=p_{eq}\frac{f_{\alpha }\left(
t\right) }{|I_{\alpha}|}
\end{equation*}%
the instantaneous pressure in $D _{\alpha }$, which is
reduced or increased according to the number of particles $N_{\alpha }(t)$
in this sub-region of interest at time $t$, and in which
$f_{\alpha }\left(t\right) = N_{\alpha }(t)/N$. We will also denote by
\begin{equation*}
\pi _{\alpha }\left( t\right) =p_{\alpha }\left(
t\right) /p_{eq}
\end{equation*}
the ratio of the instantaneous pressure in $D_{\alpha}$ to
its equilibrium value.

For the purposes of illustrating the consequences of our main theorem, suppose it
were possible to make localized
pressure measurements over millimeter length scales in a gas of the sort
described above, to a relative accuracy of $5$ ppm, i.e., $\delta_{\pi}
=5\times 10^{-6}=0.0005\%$.

Suppose, furthermore, that all $N$ particles were, e.g., by a movable rigid membrane, adiabatically compressed
to a region of the cell that contains the millimeter-cubed sub-region $D_{\alpha}$ described above, but which
itself occupies only half of the total cell volume $|D|$. The pressure in this half of the cell would therefore have increased to two atmospheres, while the pressure in the other half vanishes.
Imagine, finally, that the membrane confining the particles, initially, to this region were to suddenly burst, allowing for a free expansion of the compressed gas into the full volume $|D |$ of the cell.

Starting from an arbitrary initial condition drawn from the product distribution implied by the initial conditions described,
Theorem~\ref{thm:gasexpansion} implies that at a time $t$ greater than a relatively
short equilibration time $t_{0}$ defined in the theorem, the probability that
the relative pressure $\pi _{\alpha }\left( t\right)$ in $D_{\alpha}$ will be found to
exhibit a deviation from unity by a measurable amount $\delta_{\pi}$ is bounded from above by the
relation
\begin{equation*}
P\left( \left\vert \pi _{\alpha }\left( t\right) -1\right\vert >\delta_\pi
\right)  \leq \exp \left( -2\varepsilon ^{2}N\right) ,
\end{equation*}
in which $\varepsilon = |I_{\alpha}|\delta_{\pi} =5\times10^{-9}$.
Using this value for $\varepsilon $ we thus obtain the numerical bound
\begin{eqnarray*}
P\left( \left\vert \pi _{\alpha }\left( t\right) -1\right\vert >\delta\pi
\right)&\leq &\exp \left( -2\times (25\times 10^{-18})\times\left( 3\times 10^{19}\right) \right)  \\
&=&\exp \left( -1500\right)  \\
&<&10^{-650} .
\end{eqnarray*}

As indicated, this absurdly small bound applies to the probability of
observing such a deviation at a single instant of time $t>t_{0}.$

Consider then the probability $P_{K}$ of observing a deviation of this magnitude
at least once at some point over a quasi-continuous sequence
of $K$ times $t_{k}=t_{0}+k\Delta t,$ with $k\in \{0,\ldots ,K\}$, that
span a total time interval $\Delta T = K\Delta t$.   According to~\eqref{eq:P(goodbis)}
this quantity is bounded by the relation
\begin{equation*}
P_{K} \leq 2K\exp \left( -2\varepsilon ^{2}N\right)\leq 2K10^{-650},
\end{equation*}
where we ignored the factor $(1-\eta)$ which can at any rate be taken arbitrarily close to $1$.

Of course this upper bound on $P_{K}$ depends implicitly on the time $\Delta t$
between successive measurements and the duration $\Delta T$ of the total time
interval considered. As an extreme limit, suppose it were possible to obtain a
measurement of the pressure in this gas cell once every femtosecond
over a total time interval
$\Delta T \sim 3\times 10^{9}~\textrm{years}\sim 10^{32}~\textrm{fs}$.

Then, over the entire sequence of $K \sim 10^{32}$ such measurements, spanning
an interval of time roughly comparable to the age of the earth, the probability $P_{K}$
of observing at least one measurable deviation in the relative pressure of this
gas would be less than $2\times 10^{32}\times 10^{-650}$, giving the numerical bound
\begin{equation*}
P_{K} \leq 2\times10^{-618}.
\end{equation*}
This example shows
that the main theorem we have proven confirms the conventional wisdom regarding the
approach of individual macroscopic systems to a state of thermodynamic equilibrium, i.e.,
that it will almost certainly happen, and that once it has, the macroscopic system will almost
certainly remain in the equilibrium macrostate state for an extraordinarily long duration of time.

At the same time, the theorem itself can, in principle, be used as a
tool to provide insight into situations in which it might be possible to actually
measure deviations from the equilibrium state.

To illustrate this second aspect of our result we consider a situation similar to
that just outlined in which the initial particle gas is not at standard temperature
and pressure, but is contained in a vacuum chamber maintained at a pressure of $10^{-6}$ atm at room temperature. In that case the number of particles in
the $1\mathrm{cm}^3$ chamber would be of the order of $N= 2\times 10^{13}$. If in addition we take a more realistic value of $\delta_\pi=5\times 10^{-4}$ so that $\epsilon=5\times 10^{-7}$. One then finds
\begin{eqnarray*}
P\left( \left\vert \pi _{\alpha }\left( t\right) -1\right\vert >\delta_\pi
\right)&\leq &\exp \left( -2\times (25\times 10^{-14})\times\left( 2\times 10^{13}\right) \right)  \\
&=&\exp \left( -10\right)  \\
&<&4.5\times 10^{-5} .
\end{eqnarray*}
If one now is able to measure the pressure in this gas every second for an hour,
then with $K=3600$ one finds $$P_K\leq 2\times3600\times4.5\times10^{-5}=0.33,
$$
which suggests the possibility of detecting measurable fluctuations at
the rate of one every three hours.

The presence of the factor $\epsilon^2N$ in the exponent appearing in our main results clearly indicates that fluctuations that regularly occur in the gas density are of order $1/\sqrt N$, as expected. Indeed, this can already be observed in Fig.~2. The first example above illustrates that such small fluctuations are not easily measured in a macroscopic system of this sort under normal ambient circumstances. The second example suggests that it should be possible to produce laboratory conditions in which $\epsilon$ is indeed of the order $1/\sqrt N$. Fluctuations of this sort should, therefore, be observable.

We note in passing that experimentalists working on cold atom systems are regularly able to generate atomic gases containing on the order of $10^{9}$ atoms, confined to millimeter-cubed sized regions. Moreover, it is also possible through the application of magnetic fields to tune the inter-particle scattering length of some atomic species to zero, thus producing a gas of effectively non-interacting particles. The local particle density $n_{0} (t)$ in such a gas can then be monitored through fluctuations in the intensity of a finely collimated laser beam passing through it.

\subsection{Generalizations}
 As already pointed out in the introduction, because our proofs use very little information on the one-particle dynamics, Theorem~\ref{thm:gasexpansion} is easily generalized to other geometries. As a first set of examples, suppose that the particles move in a configuration space $M$, which can be a bounded set with boundary in $\R^d$ or a manifold. In the first case we suppose the particle executes a billiard motion, in the second that it moves along geodesics. We write $t\to (x(t), p(t))$ for the corresponding particle trajectories. Given a one-particle probability measure $\mu$ on the corresponding one-particle phase space $\Gamma_1=T^*M\simeq M\times\R^d$, suppose one can prove there exists a probability  measure $\lambda$ on $M$ such that
\begin{equation}
\lim_{t\to+\infty}\int_{\Gamma_1} \chi_I(x(t),p(t))\ \rd \mu=\lambda(I).
\end{equation}
This is the equivalent of~\eqref{eq:sortamixing}.
Once such a result is obtained, the rest of the argument leading to Theorem~\ref{thm:gasexpansion} goes through unaltered, with $|I|$ replaced by $\lambda(I)$.

This will be the case, for example, if $M$ is a negatively curved manifold of finite volume, so that the geodesic flow is mixing with respect to the Liouville measure on an energy surface. One can then consider any initial probability measure $\mu$ which is absolutely continuous with respect to the latter, and the probability measure $\lambda$ is then the normalized Riemannian volume element on $M$. Note that it  does therefore not depend on the initial probability measure $\mu$. In other words, particles initially distributed according to $\mu$  in phase space, spatially homogenize according to the measure $\lambda$ on $M$, with overwhelming probability and for long periods of time.

A similar situation occurs if $M$ is a chaotic billiard. It also occurs when $M$ is a cubic billard: the standard trick of ``unfolding'' single trajectories (see for example~\cite{Be10}) reduces this case to straight line motion on a torus, which is explicitly treated in this paper.

But other measures $\lambda$ may arise that do depend on $\mu$. For example, if $M$ is a disk in two dimensions, and if we take $\mu$ supported close to the boundary, with a distribution of momenta such that there is a minimal impact parameter, then none of the particles will enter a small disk close to the origin. This will not happen with a physical ideal gas, of course, since then the collisions of the gas particles with each other will change their direction. Such collisions have not been accounted for in the model considered here and become critically important in, e.g., the situation in which an external potential is added to the one particle Hamiltonian. Consider for example noninteracting particles bouncing back and forth in the closed interval $[0,1]$, subject to a constant force $F>0$. Such particles will spend less time close to the right edge of the interval than to the left edge, since they move faster on the right than on the left. This is clearly not what happens in a physical gas where interparticle collisions allow for equipartition of the energy. Indeed, in a physical gas
the mean speed of the particles is constant throughout, and particles spend more time in the region of lower potential energy, on the right rather than on the left. Results of the present paper, obtained for independent, non-interacting particles, clearly rely heavily on the statistical independence of those particles, which implies a law of large numbers. Thus, while they provide useful insight into the mechanism of spatial homogenization for this class of models, it is not clear that similar techniques could be usefully applied to the much more difficult, physically relevant case, in which the particles interact and exchange energy.

\section{The Kac ring model}\label{s:Kacring}
The Kac ring model, introduced in~\cite{Ka59}, describes the dynamics of $N$ particles placed on $N$ equidistant sites on a circle, one particle per site. There are, at each integer time $t$, $N_b(t)$ black balls and $N_w(t)$ white balls, so that $N=N_b(t)+N_w(t)$. Among the $N$ sites, $0\leq m\leq N$ are marked. The system undergoes a discrete-time deterministic dynamics, defined as follows. At each integer time $t$, all balls move one step counterclockwise. A ball will change color if it starts on a marked site, and not otherwise. For all $0\leq n<N$, $\xi_n=-1$ if site $n$ is marked, otherwise $\xi_n=1$. Also, $\eta_n(t)=1$ if at time $t$ the ball on site $n$ is white, and $\eta_n(t)=-1$ otherwise. One therefore has the following dynamics
$$
\eta_n(t)=\xi_{n-1}\eta_{n-1}(t-1).
$$
For the sake of the discussion here, we will suppose $\eta_n(0)=1$, so all balls are white initially. The indices are taken modulo $N$.
We are interested in the evolution of the macroscopic variable
$$
\Delta(t)=N_w(t)-N_b(t), \quad \overline{\Delta(t)}=\frac{\Delta(t)}{N}.
$$
We have
\begin{equation}\label{eq:Xndef}
\Delta(t)=\sum_{n=0}^{N-1} X_{n,t}, \quad\mathrm{where}\quad X_{n,t}=\xi_{n-1}\xi_{n-2}\dots \xi_{n-t}.
\end{equation}
We will study this quantity for choices of marked sites constructed as follows. We consider the $\xi_n$ as independent Bernoulli random variables
with, for some $0<\mu\leq 1$,
\begin{equation}
P(\xi_n=-1)=\mu,\quad P(\xi_n=1)=1-\mu.
\end{equation}
We will show that with an overwhelmingly large probability, and for $t$ within an appropriate $N$-dependent time range, the random variable $\Delta(t)$ tends to zero. This generalizes, as we will explain, known results on the model, recalled below and proven in~\cite{Ka59}; see also~\cite{GoOl09}. We first remark that, provided $0\leq t\leq N$,
$$
\E(\overline{\Delta(t)})=(1-2\mu)^t,
$$
since for $1\leq t\leq N$,
\begin{equation}\label{eq:Xnexpected}
\E(X_{n,t})=\E(\xi_{n-1}\xi_{n-2}\dots \xi_{n-t})=(1-2\mu)^t.
\end{equation}
This follows since all factors $\xi_k$ occurring in the product are independent. Hence, in analogy with~\eqref{eq:convL1}
\begin{equation}\label{eq:convL1Kac}
\lim_{\stackrel{N,t\to+\infty}{t<N}}\E\left(\overline{\Delta(t)}\right)=0.
\end{equation}
In addition, it is proven in~\cite{Ka59, GoOl09}, that
\begin{equation}\label{eq:Kacvariance}
\V(\overline{\Delta(t)})\leq \frac{C_\mu}{N}, \quad \forall t<\frac{N}2.
\end{equation}
Hence, in analogy with~\eqref{eq:convL2}, we have here
\begin{equation}\label{eq:convL2Kac}
\lim_{\stackrel{N,t\to +\infty}{t<\frac{N}2}}\E\left(\left(\overline{\Delta(t)}\right)^2\right)=0.
\end{equation}
So, in the Kac ring model, the difference between the number of white and black balls, as a fraction of the total number of balls, tends on average to zero for large $N$ and $t$.  These statements are the analogues of~\eqref{eq:gasexpansionmean}-\eqref{eq:convL2} in the expanding gas. Note however the restriction $t<N$ in~\eqref{eq:convL1Kac}, which is a consequence here of the fact that $\Delta(N)=(-1)^m\Delta(0)$, where $m$ is the number of markers. Hence $\Delta(2N)=\Delta(0)$, so that the system is time-periodic, with the same period for all configurations of markers. In addition, the variance of $\overline{\Delta(t)}$ tends to zero on the slightly shorter time scale $t<N/2$. As in the case of the expanding gas, these two statements do not suffice to say that ``typical'' configurations of markers will result in the system converging to the equilibrium value, meaning $\overline{\Delta(t)}\to 0$ with ``overwhelming'' probability. Of course, Markov's inequality together with the bound in~\eqref{eq:Kacvariance} can be used to yield a weak law of large numbers and a lower bound on this probability, but only one algebraic in $N$. As an example of such an approach, it is proven in~\cite{MaNeSh09} that, for all $\epsilon>0$ and  $T$ fixed,
$$
\lim_{N\to+\infty}P(\forall 0\leq t\leq T, |\overline{\Delta(t)}-(1-2\mu)^t|\leq \epsilon)=1.
$$
We now show how to apply a further result of~\cite{Ho63} to obtain a (sub-) exponential estimate.
For that purpose, we introduce, as before, for each $\epsilon>0$, for each $t$ and $N$,  the ``good'' set of markers $\xi_1,\dots, \xi_N$,
\begin{equation}\label{eq:Kacgoodsets}
\Omega(\epsilon, N, t)=\{(\xi_1, \dots, \xi_N)\in\{-1,1\}^N\mid\, \mid\overline{\Delta(t)}\mid\leq \epsilon\}
\end{equation}
for which the difference between the number of white and black balls, as a fraction of $N$, is within $\epsilon$ from its equilibrium value $0$ at the time $t$. Also
\begin{equation}
\Omega_T=\cap_{t=1}^T \Omega(\epsilon, N, t).
\end{equation}
Our main result is then:
\begin{theorem}\label{thm:KacRing} Let $\epsilon>0$ and $0<\alpha<1$. Let $N_\epsilon=\left(\frac{4}{\epsilon}\right)^{\frac{1}{1-\alpha}}$, $t_0=(\epsilon/4)\ln|1-2\mu|$ and $t_N=\frac12N^{\alpha}$. Then, for all $N>N_\epsilon$, for all $t_0\leq t\leq t_N$, one has
\begin{equation}\label{eq:bound2}
P(\Omega(\epsilon, N, t))\geq 1-2\exp(2(\frac{\epsilon}{2})^2)\exp(-2(\frac{\epsilon}{2})^2N^{1-\alpha}).
\end{equation}
Hence
\begin{equation}
P(\Omega_{t_N})\geq 1- \exp(2(\frac{\epsilon}{2})^2)N^{\alpha}\exp(-2(\frac{\epsilon}{2})^2N^{1-\alpha}).
\end{equation}
\end{theorem}
Note that the time scale $t_N$ is now only a power law in $N$, and not exponentially large. This is inherent in the model, which is, as remarked above, $2N$-periodic. So Poincar\'e recurrences occur trivially at $t=2N$ and hence the time over which the system stays in macroscopic equilibrium cannot be longer than $N^\alpha$, for some $\alpha<1$. Note also that the estimate on the probability of the ``good'' initial conditions gets worse as $\alpha$ approaches $1$.
\begin{proof}  Given $t,N$, we define $k,r\in\N$ by $N=kt+r$, with $0\leq r <t$. It follows from the definition of the $X_{n,t}$ in~\eqref{eq:Xndef} that $X_{n,t}$ and $X_{n', t}$ are independent provided $\|n-n'\|\geq t$, where the distance $\|n-n'\|$ between $n$ and $n'$ is measured modulo $N$. We  then rewrite $\Delta(t)$ as follows:
\begin{equation*}
\Delta(t)=\Delta_0(t)+R(t),
\end{equation*}
with
\begin{eqnarray*}
\Delta_0(t)&=&\sum_{i=1}^{t} S_{i,t},\, \mathrm{where}\, S_{i,t}=\sum_{j=0}^{k-1} X_{i+jt,t},\\
R(t)&=&\sum_{n=kt+1}^N X_{n,t}.
\end{eqnarray*}

Note that, in view of what precedes, each $S_i$ is a sum of $k$ independent random variables and that
\begin{equation}\label{eq:Rbound}
\mid \Delta(t)-\Delta_0(t)\mid=\mid R(t)\mid \leq r<t.
\end{equation}
It follows that
\begin{eqnarray}
|\frac1{N}\Delta(t)-\frac{1}{kt}\Delta_0(t)|&\leq& |\frac{1}{N}\Delta(t)-\frac1{N}\Delta_0(t)|+|\frac1{kt}\Delta_0(t)|(1-\frac{kt}{N})|\nonumber\\
&\leq& 2\frac{r}{N}\leq \frac{2t}{N}.\nonumber
\end{eqnarray}
With $N_\epsilon=\left(\frac{4}{\epsilon}\right)^{\frac{1}{1-\alpha}}$ and $t_N=\frac12N^\alpha$, for some $0<\alpha<1$,
we then have, for all $N>N_\epsilon$ and all $t<t_N$,
\begin{equation}\label{eq:bound}
|\frac1{N}\Delta(t)-\frac{1}{kt}\Delta_0(t)|\leq \frac{\epsilon}{4}.
\end{equation}
We write
$$
\overline{\Delta_0(t)}=\frac{1}{kt}\Delta_0(t).
$$
The central ingredient of the proof is the following bound, which is a variation of a result of~\cite{Ho63}. For all $\epsilon>0$, for all $N$, for all $t<t_N$,
\begin{equation}\label{eq:Kacconcentration}
P(|\overline{\Delta_0(t)}-(1-2\mu)^t|\leq \frac{\epsilon}{2})\geq 1-2\exp(2(\frac{\epsilon}{2})^2)\exp(-2(\frac{\epsilon}{2})^2N^{1-\alpha}).
\end{equation}
Before proving~\eqref{eq:Kacconcentration}, we first show how it implies the result. Combining~\eqref{eq:Kacconcentration} with~\eqref{eq:bound}, and assuming in addition that $N>N_\epsilon$, one obtains
\begin{equation*}\label{eq:bound1}
P(|\overline{\Delta(t)}-(1-2\mu)^t|\leq \frac{3\epsilon}{4})\geq 1-2\exp(2(\frac{\epsilon}{2})^2)\exp(-2(\frac{\epsilon}{2})^2N^{1-\alpha}).
\end{equation*}
Now, if $t_0\leq t\leq t_N$, where $|1-2\mu|^{t_0}=\epsilon/4$, ($t_0=\epsilon/4\ln|1-2\mu|$), we find that
\begin{equation*}\label{eq:bound2}
P(|\overline{\Delta(t)}|\leq \epsilon)\geq 1-2\exp(2(\frac{\epsilon}{2})^2)\exp(-2(\frac{\epsilon}{2})^2N^{1-\alpha}),
\end{equation*}
which is the desired result.

So it remains to prove~\eqref{eq:Kacconcentration}. Note that $\overline{\Delta_0(t)}$ can be rewritten as follows:
$$
\overline{\Delta_0(t)}=\frac{1}{t}\sum_{i=1}^t \overline{S_{i,t}},\quad \overline{S_{i,t}}=\frac1{k}S_{i,t},
$$
It follows moreover from~\eqref{eq:Xnexpected} that, for $0\leq t\leq N$,
$$
\E(\overline{\Delta_0(t)})=(1-2\mu)^t=\E(\overline{S_{i,t}}).
$$
Since the exponential function is convex, we have furthermore, for all $a\in\R$,
$$
\exp(a\left[\overline{\Delta_0(t)}-(1-2\mu)^t\right])\leq \frac1{t}\sum_{i=1}^t\exp(a\left[\overline{S_{i,t}}-(1-2\mu)^t\right])
$$
and hence
\begin{equation*}
\E(\exp(a\left[\overline{\Delta_0(t)}-(1-2\mu)^t)\right])\leq \frac1{t}\sum_{i=1}^t\E(\exp(a\left[\overline{S_{i,t}}-(1-2\mu)^t\right]).
\end{equation*}
Consequently, the Chernoff inequality reads, for all $a\in\R$,
\begin{eqnarray*}
P(\overline{\Delta_0(t)}-(1-2\mu)^t\geq \epsilon)&\leq& \exp(-a\epsilon)\E(\exp(a\left[\overline{\Delta_0(t)}-(1-2\mu)^t\right])\\
&\leq&\exp(-a\epsilon)\left(\frac1{t}\sum_{i=1}^t\E(\exp(a\left[\overline{S_{i,t}}-(1-2\mu)^t\right])\right)\\
&\leq&\exp(-a\epsilon)\left[\E(\exp(\frac{a}{k}\left[X_{0,t}-(1-2\mu)^t\right]))\right]^k.\\
\end{eqnarray*}
Here we used in the last line the fact that the $S_{i,t}$ are sums of $k$ independent random variables, all with the same mean. According to equation (4.16) in~\cite{Ho63}, one has
\begin{equation}\label{eq:hoeffdingineq}
\E(\exp(\frac{a}{k}(X_{0,t}-(1-2\mu)^t))\leq \exp(\frac18 \frac{a^2}{k^2}).
\end{equation}
Combining the last two inequalities, one finds, for all $a\in\R$,
\begin{equation}
P(\overline{\Delta_0(t)}-(1-2\mu)^t\geq \epsilon)\leq \exp(-a\epsilon+\frac18\frac{a^2}{k}).
\end{equation}
Hence, setting $a=4k\epsilon$ in the right hand side, we find
$$
P(\overline{\Delta_0(t)}-(1-2\mu)^t\geq \epsilon)\leq \exp(-2k\epsilon^2).
$$
Since $t<t_N=\frac12N^\alpha$, for some $0<\alpha<1$, we see that $r<{t}< N^\alpha$ and
$$
kN^\alpha> kt=N-r\geq N- N^\alpha,
$$
so that
$$
k\geq N^{1-\alpha}-1.
$$
Hence, for all $N$ and for all $t<t_N$
$$
P(\overline{\Delta_0(t)}-(1-2\mu)^t\geq \epsilon)\leq \exp(2\epsilon^2)\exp(-2\epsilon^2N^{1-\alpha}).
$$
This implies~\eqref{eq:Kacconcentration}.
\end{proof}

\section{Discussion}\label{s:discussion}
We now turn to the question of why the results established above are not in contradiction with the Poincar\'e recurrence theorem. In other words, why is there no problem with Zermelo's paradox in the present context?

We illustrate the issue by considering a one-dimensional gas of $N$ particles having an initial probability measure $\rd \mu(x,p)=\rho(x,p)\rd x\rd p$ that is absolutely continuous with respect to Lebesgue measure. For such a gas there will be initial states that display obvious recurrences in the form of periodic orbits whenever $p_{0 1}=p_{0 2}=\dots=p_{0 N}$, at times $T_n=n/|p_{0 1}|$, $n\in\N$. Choosing, in addition, $x_{01}=x_{02}=\dots=x_{0N}$ with  $0.2\leq x_{01}\leq 0.3$ and $I=[0.5, 1]$, it is then clear that $f_I(X_0+T_n P_0, P_0)=0$, very far from the equilibrium  value $0.5$. It furthermore follows that for all initial conditions $(X', P')$ sufficiently close to $(X_0, P_0)$, one has that $(X'+T_n P', P')$ is still close to $(X_0,P_0)$, so that again $f_I(X'+T_n P', P')=0$. Let us write $B_{N,n}$ for these open sets of initial conditions. They are a shrinking family of sets, as a function of $n$ at fixed $N$, and have an exponentially decreasing Lebesgue measure, as a
function of $N$ at fixed $n$.

Now, suppose that, in addition to $N$, the single particle distribution $\rho$ is also specified and that $t_0<t_N$, where $t_0$ and $t_N$ are defined in Theorem~\ref{thm:gasexpansion}. Remark that, in view of the definition of $t_0$, this constitutes a condition linking $\rho$ to $N$. At fixed $N$, the distribution $\rho$ cannot be too concentrated in the momentum variables of the particles. With $\rho$ fixed, $N$ has to be sufficiently large. Under these circumstances, for each $n\in\N$, either $T_n<t_0$, or $t_0\leq T_n\leq t_N$, or $t_N<T_n$.

In the first case, which may occur when $n$ is small, there is obviously no contradiction: equilibrium has not yet been reached by the $N$ particle gas considered. In the third case, which occurs for all sufficiently large $n$, equilibrium is no longer guaranteed by Theorem~\ref{thm:gasexpansion}, so that periodic recurrences do not represent a problem. In the second case, $B_{N, n}$  is necessarily a subset of the ``bad'' set of initial conditions $\Omega^c_{K_N}$ and the theorem asserts that its probability is therefore exponentially small in $N$, a fact we already knew in this particularly simple situation since the Lebesgue measure of $B_{N, n}$ is exponentially small in $N$. In other words, given $\rho$, $N$ and $n$, the fact that $t_0<T_n<t_N$ implies that $\rho$ varies sufficiently slowly in the momentum variable for the set $B_{N, n}$ of bad initial conditions to have a very small probability. To put it another way, if we fix $n$ and choose $\rho$ so that the probability of $B_{N, n}$ is large, then necessarily $T_n\leq t_0$.

The analysis can be repeated for Poincar\'e recurrences not corresponding to periodic orbits as well. A similar analysis shows there is also no contradiction with Proposition~\ref{lem:L1L2conv} due to the essential fact that the estimates~\eqref{eq:control1} and~\eqref{eq:control2}, like the equilibration time $t_0$ in Theorem~\ref{thm:gasexpansion}, depend on $\rho$.

For the Kac ring model, on the other hand, the statement of our results limits the times considered to $t<< 2N$, thereby avoiding the global period of the dynamics at $t=2N$.

To see why the irreversibility which is manifest in the behaviour of the approach of the gas to a flat density profile is not in contradiction with the reversible nature of the underlying microscopic dynamics (Loschmidt's paradox), consider the following. Suppose $I=[0.5, 1]$, $\epsilon=0.01$, and that $N$ and $\rho$ have been fixed. Suppose the $x$-support of $\rho$ is inside $[0.1, 0.2]$. Consider a time $t_0<T<t_N$ and an initial condition $(X_0, P_0)$ inside $\Omega_{K_N}$, with $x_i\in[0.1, 0.2]$: then $|f_I(X_0+P_0T, P_0)-0.5|\leq \epsilon=0.01$,  so that the gas is close to equilibrium at time $T$. Now consider the new initial condition $(X'_0=X_0+P_0T, P'_0=-P_0)$, which, after the same time $T$ evolves back to the microstate $(X_0, P_0)$. Clearly, since $f_I(X'_0, P'_0, T)=f_I(X_0, P_0, 0)=0$, the state $(X'_0, P'_0)$ must belong to the bad set of initial conditions $\Omega_{K_N}^c$. In other words, if one evolves the good set of initial conditions $\Omega_{K_N}$ over a time $T$, then reverses all velocities, one obtains a subset $\widehat\Omega_{K_N, T}$ of the bad set of initial conditions $\Omega^c_{K_N}$. It must therefore have a very small probability, and the initial conditions inside are, in this sense, ``atypical''.

This may seem surprising since the set $\widehat\Omega_{K_N, T}$ has the same Lebesgue measure as $\Omega_{K_N}$ itself, by the Liouville theorem. But this does not contradict the fact that it has an exceedingly small probability with respect to the probability induced on $\Gamma_N$ by the density $\rho$.

Finally, we observe the following. Suppose we introduce a new probability density$\rho'$  on $\Gamma_1$, defined by $\rho'(x',p')=\rho(x'+p'T,-p')$. Then, with respect to this new probability, the set $\widehat\Omega_{K_N, T}$ has the same large probability as does the set $\Omega_{K_N}$ with respect to the probability induced by the original $\rho$. Nevertheless, for all initial conditions $(X'_0, P'_0)\in\widehat\Omega_{K_N, T}$, one has by what precedes that $f_I(X'_0+P'_0T, P'_0)=0$; so the gas is far from equilibrium after a time $T$.

This again may seem surprising, since the time $T$ was chosen larger than the equilibration time $t_0$. One should however remember that the latter depends on $\rho$, and that the equilibration time $t'_0$  of the new $\rho'$ is much longer than that of the original $\rho$, as is readily checked. So, one has to evolve the gas much longer for equilibrium to be attained, if starting from an initial condition in $\widehat\Omega_{K_N, T}$. The point is, in the end, that whereas the sets of ``good'' initial conditions $\Omega_I(\epsilon, N, t)$ do not, by themselves, depend on $\rho$, their probability, of course, does. Hence the typicality of the initial conditions in $\Omega_I( \epsilon, N, t)$ depends on the choice of $\rho$. Recall that $\rho$ is used only to describe the preparation of the initial state of the gas: initially, the gas particles are independently distributed in the one-particle phase space $\Gamma_1$ according to the density $\rho$.

We finally remark that no correlations are created by the evolution of the free gas we consider here, since the particles in this gas model do not interact. Loschmidt's paradox is therefore not resolved here by saying that the time-evolved, time-reversed initial ``typical'' conditions becomes ``atypical'' because of correlations. It is simply the result of the faster oscillations one finds in $\rho'$ than in $\rho$, which increase the equilibration time.
Another way to put it is that, since the initial probability we choose on $\Gamma_N$ is a product, the particles are initially and at all times independent. Thus, the Stosszahlansatz is trivially satisfied here at all times. In other words, as it is sometimes put, molecular chaos is trivially propagated~\cite{Vi12}.


\end{document}